\documentclass[11pt]{article}

\usepackage[textwidth=1.8cm,textsize=scriptsize]{todonotes}

\usepackage[utf8]{inputenc}
\usepackage[T1]{fontenc}
\usepackage{lmodern}
\usepackage[DIV=11]{typearea} % large value = less margin (roughly speaking), recommended values: 10pt: 8, 11pt: 10, 12pt: 12

\usepackage{microtype}

\usepackage{amsmath}
\usepackage{amssymb}
\usepackage{amsthm}
\usepackage{thmtools}

\usepackage[procnumbered,ruled,vlined,linesnumbered]{algorithm2e}
 
\SetCommentSty{mycommfont} % set comment style for algorithms (e.g. when use \tcp)

\usepackage{xcolor}
\usepackage{xspace}
\usepackage{enumitem}
\usepackage{tikz}

\usepackage{comment}

\usepackage{graphicx}
\usepackage{caption}
\usepackage{subcaption}

\usepackage{authblk}

\usepackage{microtype} % Magically save some space

\def\AdvCite{True} % SWITCH advanced citation on or off

\ifdefined\AdvCite

\usepackage[
	style=alphabetic,
	backref=true,
	doi=false,
	url=false,
	maxcitenames=3,
	mincitenames=3,
	maxbibnames=10,
	minbibnames=10,
	backend=bibtex8,
	sortlocale=en_US
]{biblatex}

\usepackage[ocgcolorlinks]{hyperref} % Option ocgcolorlinks makes links only colored on screen and not on printouts
\usepackage{cleveref}

\DeclareCiteCommand{\citem}
    {}
    {\mkbibbrackets{\bibhyperref{\usebibmacro{postnote}}}}
    {\multicitedelim}
    {}

\renewcommand*{\multicitedelim}{\addcomma\space}

\AtEveryCitekey{\clearfield{note}}

\newcommand{\myhref}[1]{%
	\iffieldundef{doi}
	{\iffieldundef{url}
		{#1}
		{\href{\strfield{url}}{#1}}}
	{\href{http://dx.doi.org/\strfield{doi}}{#1}}%
}

\DeclareFieldFormat{title}{\myhref{\mkbibemph{#1}}}
\DeclareFieldFormat
[article,inbook,incollection,inproceedings,patent,thesis,unpublished]
{title}{\myhref{\mkbibquote{#1\isdot}}}

\makeatletter
\AtBeginDocument{%
	\newlength{\temp@x}%
	\newlength{\temp@y}%
	\newlength{\temp@w}%
	\newlength{\temp@h}%
	\def\my@coords#1#2#3#4{%
		\setlength{\temp@x}{#1}%
		\setlength{\temp@y}{#2}%
		\setlength{\temp@w}{#3}%
		\setlength{\temp@h}{#4}%
		\adjustlengths{}%
		\my@pdfliteral{\strip@pt\temp@x\space\strip@pt\temp@y\space\strip@pt\temp@w\space\strip@pt\temp@h\space re}}%
	\ifpdf
	\typeout{In PDF mode}%
	\def\my@pdfliteral#1{\pdfliteral page{#1}}% I don't know why % this command...
	\def\adjustlengths{}%
	\fi
	\ifxetex
	\def\my@pdfliteral #1{}% isn't equivalent to this one
	\def\adjustlengths{\setlength{\temp@h}{-\temp@h}\addtolength{\temp@y}{1in}\addtolength{\temp@x}{-1in}}%
	\fi%
	\def\Hy@colorlink#1{%
		\begingroup
		\ifHy@ocgcolorlinks
		\def\Hy@ocgcolor{#1}%
		\my@pdfliteral{q}%
		\my@pdfliteral{7 Tr}% Set text mode to clipping-only
		\else
		\HyColor@UseColor#1%
		\fi
	}%
	\def\Hy@endcolorlink{%
		\ifHy@ocgcolorlinks%
		\my@pdfliteral{/OC/OCPrint BDC}%
		\my@coords{0pt}{0pt}{\pdfpagewidth}{\pdfpageheight}%
		\my@pdfliteral{F}% Fill clipping path (the url's text) with
		\my@pdfliteral{EMC/OC/OCView BDC}%
		\begingroup%
		\expandafter\HyColor@UseColor\Hy@ocgcolor%
		\my@coords{0pt}{0pt}{\pdfpagewidth}{\pdfpageheight}%
		\my@pdfliteral{F}% Fill clipping path (the url's text)
		\endgroup%
		\my@pdfliteral{EMC}%
		\my@pdfliteral{0 Tr}% Reset text to normal mode
		\my@pdfliteral{Q}%
		\fi
		\endgroup
	}%
}
\makeatother

\else

\usepackage[ocgcolorlinks]{hyperref} % Option ocgcolorlinks makes links only colored on screen and not on printouts
\usepackage{cleveref}

\fi

\allowdisplaybreaks[1]

\makeatletter
\g@addto@macro\bfseries{\boldmath}
\makeatother

\makeatletter
\g@addto@macro\mdseries{\unboldmath}
\g@addto@macro\normalfont{\unboldmath}
\g@addto@macro\rmfamily{\unboldmath}
\g@addto@macro\upshape{\unboldmath}
\makeatother

\renewcommand{\paragraph}[1]{\medskip\noindent{\bf #1}\xspace}

\colorlet{DarkRed}{red!50!black}
\colorlet{DarkGreen}{green!50!black}
\colorlet{DarkBlue}{blue!50!black}

\hypersetup{
	linkcolor = DarkRed,
	citecolor = DarkGreen,
	urlcolor = DarkBlue,
	bookmarks = true,
	bookmarksnumbered = true,
	linktocpage = true
}

\declaretheorem[numberwithin=section]{theorem}
\declaretheorem[numberlike=theorem]{lemma}

\declaretheorem[numberlike=theorem]{corollary}
\declaretheorem[numberlike=theorem]{definition}
\declaretheorem[numberlike=theorem]{claim}
\declaretheorem[numberlike=theorem]{observation}

\crefname{algorithm}{Algorithm}{Algorithms}
\Crefname{algorithm}{Algorithm}{Algorithms}

\DontPrintSemicolon
\SetKw{KwAnd}{and}
\SetProcNameSty{textsc}
\SetFuncSty{textsc}

\newcommand{\ot}{\tilde{O}}

\newcommand{\ignore}[1]{}
\newcommand{\distthrough}{{\sf DistThrough}}

\newcommand{\dist}{{\sf dist}}
\renewcommand{\sp}{\pi}
\newcommand{\hop}{{\sf hop}} 
\newcommand{\distest}{\hat{d}}
\newcommand{\diam}{{D}\xspace}

\newcommand{\congest}{\textsf{CONGEST}\xspace}

\newcommand{\poly}{\operatorname{poly}}
\newcommand{\polylog}{\operatorname{polylog}}

\newcommand{\congestion}{{\sf congestion}\xspace}
\newcommand{\dilation}{{\sf dilation}\xspace}

\newcommand{\FilteredBroadcast}{\mbox{{\sc RandFilteredBroadcast}}\xspace} 
\newcommand{\filteredbroadcast}{\FilteredBroadcast} 

\ifdefined\ShowComment

\def\danupon#1{\marginpar{$\leftarrow$\fbox{D}}\footnote{$\Rightarrow$~{\sf #1 --Danupon}}}
\def\aaron#1{\marginpar{$\leftarrow$\fbox{A}}\footnote{$\Rightarrow$~{\sf #1 --Aaron}}}

\else

\def\danupon#1{}
\def\aaron#1{}

\fi

\title{Distributed Exact Weighted All-Pairs Shortest Paths in Near-Linear Time} 

\author[1]{Aaron Bernstein}
\author[2]{Danupon Nanongkai}
\affil[1]{Rutgers University. {\tt bernstei@gmail.com}}
\affil[2]{KTH Royal Institute of Technology, Sweden. {\tt danupon@gmail.com}}

\date{}

\hypersetup{
	pdftitle = {Distributed Exact Weighted All-Pairs Shortest Paths in Near-Linear Time.},
	pdfauthor = {Aaron Bernstein and Danupon Nanongkai}
}

\begin{document}

	\begin{titlepage}
		\maketitle
		\pagenumbering{roman}
		%\vspace{-.7cm}
		
\begin{abstract}
In the {\em distributed all-pairs shortest paths} problem (APSP), every node in the weighted undirected distributed network (the CONGEST model) needs to know the distance from every other node using least number of communication rounds (typically called {\em time complexity}). The problem admits $(1+o(1))$-approximation $\tilde\Theta(n)$-time algorithm and a nearly-tight $\tilde \Omega(n)$ lower bound [Nanongkai, STOC'14; Lenzen and Patt-Shamir PODC'15]\footnote{$\tilde \Theta$, $\tilde O$ and $\tilde \Omega$ hide polylogarithmic factors. Note that the lower bounds also hold even in the unweighted case and in the weighted case with polynomial approximation ratios~\cite{LenzenP_podc13,HolzerW12,PelegRT12,Nanongkai-STOC14}.}. For the exact case, Elkin [STOC'17] presented an $O(n^{5/3} \log^{2/3} n)$ time bound, which was later improved to $\tilde O(n^{5/4})$ [Huang, Nanongkai, Saranurak FOCS'17].	It was shown that any super-linear lower bound (in $n$) requires a new technique [Censor-Hillel, Khoury, Paz, DISC'17], but otherwise it remained widely open whether there exists a $\tilde O(n)$-time algorithm for the exact case, which would match the best possible approximation algorithm. 
	
This paper resolves this question positively: we present a randomized (Las Vegas) $\tilde O(n)$-time algorithm, matching the lower bound up to polylogarithmic factors. Like the previous $\tilde O(n^{5/4})$ bound, our result works for directed graphs with zero (and even negative) edge weights. In addition to the improved running time, our algorithm works in a more general setting than that required by the previous $\tilde O(n^{5/4})$ bound; in our setting (i) the communication is only along edge directions (as opposed to  bidirectional), and (ii) edge weights are arbitrary (as opposed to integers in  $\{1, 2, \ldots, \poly(n)\}$). The previously best algorithm for this more difficult setting required $\tilde O(n^{3/2})$ time [Agarwal and Ramachandran, ArXiv'18]  (this can be improved to $\tilde O(n^{4/3})$ if one allows bidirectional communication).

Our algorithm is extremely simple and relies on a new technique called {\em Random Filtered Broadcast}. Given any sets of nodes $A,B\subseteq V$ and assuming that every $b \in B$ knows all distances from nodes in $A$, and every node $v \in V$ knows all distances from nodes in $B$, we want every $v\in V$ to know  $\distthrough_B(a,v) = \min_{b\in B} \dist(a,b) + \dist(b,v)$ for every $a\in A$. Previous works typically solve this problem by broadcasting all knowledge of every $b\in B$, causing super-linear edge congestion and time. We show a randomized algorithm that can reduce edge congestions and thus solve this problem in $\tilde O(n)$ expected time.
\end{abstract}

		\newpage
		\setcounter{tocdepth}{2}
		\tableofcontents
		%\newpage
		%\listoftheorems
	\end{titlepage}

	\newpage
	\pagenumbering{arabic}

	\section{Introduction} \label{sec:intro}

%\paragraph{All-Pairs Shortest Paths (APSP).} 

We study the distributed all-pairs shortest paths problem (APSP) defined on the \congest model of distributed network. A network is modeled by a weighted undirected $n$-node graph $G=(V, E)$.\footnote{As we will discuss later, we can also handle directed graphs.} 
Each node represents a processor with unique ID and infinite computational power that initially only knows its adjacent edges and their weights. 
Nodes can communicate with each other in {\em rounds}, where in each round each node can send a message of size $O(\log n)$ to each neighbor (weights play no role in the communication). 
The goal of APSP is for every node to know its distances from all other nodes. We want an algorithm that achieves this with smallest number of rounds, called {\em time complexity}.
%\footnote{This problem is sometimes referred to as {\em name-independent routing schemes}. See, e.g. \cite{LenzenP_stoc13,LenzenP-podc15}  for discussions and results on another variant called {\em name-dependent routing schemes} which is not considered in this paper.}   
%
It is usually expressed in terms of $n$ and $\diam$, where $n$ is the numer of nodes and $\diam$ is the diameter of the network when edge weights are omitted. Throughout we use $\tilde \Theta$, $\tilde O$ and $\tilde \Omega$ to hide polylogarithmic factors in $n$. See \Cref{sec:prelim} for details of the model.

The {\em approximate} version of the problem was known to admit (i) a $(1+o(1))$-approximation $\tilde O(n)$-time deterministic algorithm and (ii) an $\tilde \Omega(n)$ lower bound which holds even against randomized $\tilde O(\poly(n))$-approximation algorithms and when $\diam=O(1)$~\cite{LenzenP_stoc13,LenzenP-podc15,Nanongkai-STOC14}. The exact {\em unweighted} version was also settled with $\tilde\Theta(n)$ bound \cite{LenzenP_podc13,HolzerW12,FrischknechtHW12,PelegRT12,AbboudCK16}.\footnote{More precisely, the bound for the unweighted case is $\Theta(n/\log n)$. The lower bound holds against $(\polylog(n))$-approximation algorithms when the network is unweighted. The same lower bound also holds even for the easier problem of approximating the network diameter~\cite{FrischknechtHW12}. }
% In particular, for the weighted case, the lower bound holds even for $(2-\epsilon)$-approximation algorithms. For the unweighted case, the lower bound holds even for $(3/2-\epsilon)$-approximating the diameter, and even for sparse networks~\cite{AbboudCK16}.}  
%
For the {\em exact weighted} case, nothing was known until the 2017 bound of $O(n^{5/3} \log^{2/3} n)$ by Elkin \cite{Elkin-STOC17}, which was later improved to $\tilde O(n^{5/4})$ \cite{HuangNS17}; both algorithms by \cite{Elkin-STOC17,HuangNS17} are randomized.  
On the lower bound side, Censor-Hillel, Khoury, and Paz \cite{Censor-HillelKP17} pushed the bound to $\Omega(n)$
%that obtaining any super-linear lower bound requires a new technique ()
and proved that the standard lower bound technique cannot provide a super-linear lower bound.  
%
%The latter remains the state of the art, leaving 
%It remains an intriguing question whether we can solve the exact weighted case as fast as the approximate and the unweighted cases. 
%
%The gap between the $\tilde O(n^{5/4})$ and $\Omega(n)$ bounds leaves 
%
Despite this, it was still widely open whether there was a new technique that implies a super-linear lower bound, or whether we can in fact solve the exact weighted case in $\tilde O(n)$ time, like the approximate and the unweighted cases. 

\paragraph{Our result.} We present an randomized (Las Vegas) $\tilde O(n)$-time algorithm. 
%This settles the exact weighted APSP problem, except for the questions about deterministic algorithms.
%It matches the lower bound up to polylogarithmic factors, and thus resolves the question positively 
This essentially settles the distributed APSP problem, with the key open remaining problem being whether deterministic algorithms can achieve the same bound. Like the previous $\tilde O(n^{5/4})$-time algorithm, our algorithm works in a more difficult model where each node must send the same message to every neighbor in each round ({\sf broadcast} \congest), and can handle a more general case of inputs: directed graphs with zero edge weights; in fact a standard reduction shows that our algorithm can also handle negative weights in $\tilde O(n)$ time.   

In addition to the improved running time, our algorithm works in a more general setting than that required by the previous $\tilde O(n^{5/4})$ bound. The previous $\tilde O(n^{5/4})$-time algorithm of \cite{HuangNS17} requires that (i) the communication is {\em bidirectional} (unaffected by edge directions), and (ii) edge weights are in $\{1, \ldots, \poly(n)\}$. 
%
%More importantly, it was mentioned in \cite{HuangNS17} that the main drawback of their technique is that the time guarantee depends on the number of bits needed to represent edge weights.
%
Although these are typical assumptions, some works have explored the possibilities to avoid them, e.g. \cite{Elkin-STOC17,AgarwalRKP18-PODC,AgarwalR18-Arxiv-3/2,AgarwalR18-Arxiv2-4/3} (the second assumption was also mentioned in \cite{HuangNS17} as their main drawback, since their guarantee depends on the number of bits needed to represent edge weights). If we do without both assumptions (so communication is only along edge directions 
%It might be impossible for $u$ to know such information, e.g. if there is no edge directed at $u$.} 
, and edge weights are arbitrary as long as a distance can be sent through a link in one round), the previously best algorithm for this more difficult setting required $\tilde O(n^{3/2})$ time \cite{AgarwalR18-Arxiv-3/2} \footnote{We emphasize that in the case of uni-directional communication, node $v$ can learn its distance from $u$ only if there is a directed path from $u$ to $v$; otherwise, it is impossible for $v$ to learn such information.}. If bidirectional communications are allowed, then the bound can be improved to $\tilde O(n^{4/3})$ \cite{AgarwalR18-Arxiv2-4/3}. 
Our algorithm does not require any of the above assumptions, and our $\tilde O(n)$ bound subsumes all above results, except that the $\tilde O(n^{3/2})$-time algorithm in \cite{AgarwalR18-Arxiv-3/2} is deterministic.

Our algorithm is also much simpler than the previous $\tilde{O}(n^{5/4})$ state-of-the-art. Given that our result is essentially optimal, we believe that its simplicity is a plus.

%Our algorithm is also \emph{much} simpler than the previous $\tilde{O}(n^{5/4})$ state-of-the-art. Given that our result is essentially optimal, we believe that its simplicity is a plus, since it results in a natural conclusion to the problem.
%To this end, we emphasize that our algorithm does not use the scaling technique, and thus does not have to limit th

\paragraph{Other related works.} As noted earlier, one aspect left to understand distributed APSP is the performance of deterministic algorithms. The current best time for deterministic algorithms is $\tilde O(n^{3/2})$, first achieved by Agarwal~et~al. \cite{AgarwalRKP18-PODC} and later tailored to work without bidirectional communication by Agarwal and Ramachandran \cite{AgarwalR18-Arxiv-3/2}. For a summary of previous algorithms and their properties, see  \cite[Table~1]{AgarwalR18-Arxiv2-4/3}. 

Distributed APSP is sometimes referred to as {\em name-independent routing schemes}. See, e.g. \cite{LenzenP_stoc13,LenzenP-podc15} for discussions and results on another variant called {\em name-dependent routing schemes} which is not considered in this paper. These papers also show an application of distributed APSP to routing tables constructions.

The previous lack of understanding for exact APSP in fact reflects a bigger issue in the field of distributed graph algorithms: Studies in the past few  years have led to tight {\em approximation} algorithms for several graph problems; for example, single-source shortest paths (SSSP), minimum cut, and maximum flow can be $(1+o(1))$-approximated in $\tilde O(\sqrt{n}+\diam)$ time~\cite{HenzingerKN-STOC16,BeckerKKL16,Nanongkai-STOC14,NanongkaiS14_disc,GhaffariK13,GhaffariKKLP15}\footnote{For the maximum flow algorithm, there is an extra $n^{o(1)}$ term in the time complexity.}, and the time bounds are tight up to polylogarithmic factors \cite{DasSarmaHKKNPPW12,Elkin06,PelegR00,KorKP13,ElkinKNP14}. In contrast, except for minimum spanning tree (e.g. \cite{KuttenP98,PanduranganRS-STOC17,Elkin-arxiv17-mst}), not much was known for {\em exact} algorithms until 2017, when algorithms for exact SSSP and APSP started to appear (e.g. \cite{GhaffariL18,ForsterN18,Elkin-STOC17,HuangNS17,AgarwalRKP18-PODC,AgarwalR18-Arxiv-3/2,AgarwalR18-Arxiv2-4/3}). Settling the exact cases for other problems remains a major open problem.

\paragraph{Techniques.}
%
%\section{For Intro: Techniques}
%
The cornerstone of our algorithm is a new technique called {\em random filtered broadcasting}. We give an overview in \cref{sec:overview}; loosely speaking, the technique applies to settings where one needs to broadcast a large amount of information to every vertex, but in the end each vertex only cares about the ``best'' message it receives. We show how to use randomization to  filter out most of the messages, and reduce the congestion on each edge. Although relatively simple, our result in this paper show the technique to be very powerful. It is also quite general, so we have strong reason to believe that it will find application in other distributed algorithms for the \congest model, especially those related to distances.

On a more concrete level, we use random filtered broadcasting to devise a primitive which leads to our APSP algorithm, but which we think may prove useful in its own right. 
%In particular, given any sets of nodes $A,B\subseteq V$ (nodes know if they are in these sets) and assuming that every $b \in B$ knows all distances from nodes in $A$, and every node $v \in V$ knows all distances from nodes in $B$, we want every $v\in V$ to know  $\distthrough_b(a,v) = \min_{b\in B} \dist(a,b) + \dist(b,v)$ for every $a\in A$. Previous works typically solve this problem by broadcasting all knowledge of every $b\in B$, causing super-linear edge congestion and time. We show a randomized algorithm that can reduce edge congestions and thus solve this problem in $\tilde O(n)$ expected time.
In particular,  
%
%
%given any sets $A,B \subseteq V$, we show that if every $b \in B$ knows all distances from $A$, and every node $v \in V$ knows all distances from $B$, then for every pair $a \in A$ and $v \in C$ we can make sure that $v$ knows $\distthrough_b(a,v) = \min_{b\in B} \dist(a,b) + \dist(b,v)$ in a total of $\ot(n)$ time. 
%
given any sets of nodes $A,B\subseteq V$ (nodes know if they are in these sets) and assuming that every $b \in B$ knows all distances from nodes in $A$, and every node $v \in V$ knows all distances from nodes in $B$, we want every $v\in V$ to know  $\distthrough_B(a,v) = \min_{b\in B} \dist(a,b) + \dist(b,v)$ for every $a\in A$.\footnote{Note that we actually have to handle a bit more general case where {\em not} all distances from $a\in A$ are known to nodes in $B$.} This was previously an obstacle for APSP.
In this paper, we show how to do this in $\ot(n)$ time. Armed with this black-box, we are able use a very natural framework for APSP.
%This process alone allows us to come up with a strikingly simple algorithm.
% without the need of complicated steps used by previous algorithms.  
%
Additionally, if we only care about hop-distances at most $h$, then we can reduce the number of rounds to $\ot(|A| + h)$. We hope that just as Bellman-Ford is often used as a primitive that allows one to separately handle shorter and longer hop-distances, our new algorithm for $\distthrough$ can be used as a primitive in other distributed shortest path algorithms.

\paragraph{Remark.} Throughout the paper we only show that the output is correct with high probability. As discussed in \cite{HuangNS17}, this can be made Las Vegas since in $\tilde O(n)$ time we can check the correctness, as follows. First, every node lets its neighbors know about its distances from other nodes (this takes $O(n)$ time). Then, every node checks if it can improve its distance from any node using the distance knowledge from neighbors. If the answer is ``no'' for every node, then the computed distance is correct. If some node answers ``yes'', it can broadcast its answer to all other nodes in $O(n)$ time.

	\section{High-Level Overview} \label{sec:overview}

We start with a randomized hierarchy: for every integer $1 \leq i \leq \log(n)$, we construct set $S_i$ by independently sampling each vertex with probability $1/2^i$; we set $S_0 = V$ and $S_{\log(n) + 1} = \emptyset$. Then with high probability: $|S_i| = \ot(n/2^i)$, and any shortest path with at least $\ot(2^i)$ vertices contains a vertex from $S_i$.

Our algorithm then proceeds in phases, following a standard framework for shortest path algorithms. We go from phase $i = \log(n)$ \emph{down to} phase $0$. The guarantee at the end of phase $i+1$ is that every vertex $v$ knows the shortest distances from each $s \in S_{i+1}$; that is, $d^v_{i+1} = \dist(s,v)$. Let us now consider phase $i$. The goal is for every node $v$ to learn all distances $\dist(s,v)$ for $s \in S_{i}$ and $v \in V$. First, each vertex $s$ in $S_i$ runs Bellman-Ford up to hop-distance $\ot(2^i)$: this gives us all distances $\dist(s,v)$ for which $\hop(s,v) = \ot(2^i)$. On the other hand, if $\hop(s,v)$ is large, then we know that there exists a vertex $s_{i+1} \in S_{i+1}$ on the shortest path $\pi(s,v)$. Note, moreover, that because of phase $i+1$ we already know $\dist(s_{i+1},v)$; it is also not hard to ensure that we know $\dist(s, s_{i+1})$ because of the Bellman-Ford computation from $s_i$. 

Thus, to complete the phase $i$, all we have left is to solve the sub-problem $\distthrough_{S_{i+1}}(S_i, V)$: we assume that we already know distances from $S_i$ to $S_{i+1}$ and from $S_{i+1}$ to $V$, and the goal is to compute $\distthrough_{S_{i+1}}(s,v) = \min_{s_{i+1} \in S_{i+1}} \dist(s,s_{i+1}) + \dist(s_{i+1},v)$ for every $s \in S_i$ and $v \in V$. (In fact the Bellman-Ford computation from each $s \in S_i$ only gives us accurate distances to some of the $S_{i+1}$, but this ends up having no effect, so for this overview we stick to the simpler description above.)

Note that \distthrough\ is a very natural problem in and of itself, and also comes up in many other shortest path algorithms. The issue is that it is not clear how to approach this problem in the distributed setting. The naive solution would be to have each $s_{i} \in S_{i+1}$ broadcast $\dist(s,s_{i+1})$ for each $s \in S_{i+1}$. But this incurs a congestion of $O(|S_i|^2)$, which is only efficient when $S_i$ is relatively small. For this reason, previous algorithms had to deviate from the simple framework described above, and typically tried to balance two different approaches, one for small-hop distances, and one for large ones; in the former case, a Bellman-Ford-style approach is efficient, while for the latter case the relevant $S_i$ is small, and so a broadcasting-type-approach is efficient. However, such a trade-off necessarily results in a super-linear round complexity, such as the state of the art of $O(n^{1.25})$.

Our main contribution is to show that $\distthrough_B(A,C)$ can be solved in $\ot(n)$ time, for \emph{any} sets $A,B,C \subseteq V$, regardless of their size. Not only does this lead to an optimal round complexity (up to log factors), but it also leads to a very clean and simple solution to the problem, as we are able to use the  framework described above, without needing to balance multiple different approaches. 

\paragraph{Random Filtered Broadcasting:} We solve $\distthrough_B(A,C)$ by using a new technique that we refer to as random filtered broadcasting. We focus on a fixed $a \in A$, and show how to solve $\distthrough_B(a,C)$ with only $\ot(1)$ congestion on each edge; using \cref{thm:Ghaffari}, we can then parallelize the algorithms for all $a \in A$ in time $\ot(|A|) = \ot(n)$. Let us consider the naive {\em broadcasting} approach again: each vertex in $b \in B$ knows all distances from $A$, so it sends a message $M(a,b) = (a, \dist(a,b))$ for every $a \in A$. Whenever a vertex $c \in C$ receives message $M(a,b)$, it can use its knowledge of $\dist(b,c)$ to compute $\distthrough_b(a,c)$. Thus, if a vertex $c$ receives $M(a,b)$ for all $b \in B$ it can compute $\distthrough_B(a,c)$.

To reduce the congestion on each edge, we allow vertices to filter out certain message $M(a,b)$, i.e. to not pass them on to their neighbors. Consider the following {\em filtering heuristic}: if a vertex $v$ sees a message $M(a,b')$, but $v$ has previously seen a message $M(a,b)$ with $\distthrough_b(a,v) \leq \distthrough_{b'}(a,v)$, then $v$ does not pass on the message $M(a,b')$. This reduces the total number of messages sent, and each $c \in C$ still correctly computes $\distthrough_B(a,c)$; the reason is that if $b$ is the vertex in $B$ that minimizes $\dist(a,b) + \dist(b,c)$, then it is not hard to see that every node on $\pi(b,c)$ will pass on message $M(a,b)$ (or some equivalently good message, in case of a tie.)

 Unfortunately in the worst-case the congestion might be no better than before, as each vertex $v$ might receive the messages $M(a,b)$ in the worst possible order -- that is, in decreasing order of $\dist(a,b) + \dist(b,v)$; in this case, $v$ will pass on every message it sees. To overcome this, we use a randomized filter. We let $B_0 = B$, and obtain each $B_j$ by sampling each node in $B_{j-1}$ with probability $1/2$. Our algorithm then proceeds in iterations, starting from $j = \log(n)$ down to $j=0$. In iteration $j$, we broadcast all message $M(a,b)$ for $b \in B_{j}$; however, as in the above paragraph, a vertex $v$ filters out messages unless they are strictly better than all previous messages $M(a,b')$ seen by $v$ -- i.e. unless $\distthrough_b(a,v)$ is smaller. The basic argument is that with high probability, $v$ will filter out all but $O(\log(n))$ messages in iteration $j$; the reason is that if we look at the $O(\log(n))$ $b \in B_j$ that are ``best'' for $v$, then with high probability at least one of them is in $B_{j+1}$, and so was already seen iteration in $j+1$, and will filter out all messages not in the top $O(\log(n))$. We thus have a total congestion of $O(\log(n))$ per iteration, and so $O(\log^2(n))$ congestion to compute $\distthrough_B(a,C)$, and $\ot(n)$ time for $\distthrough_B(A,C)$.

	\section{Preliminaries}\label{sec:prelim}

\subsection{The CONGEST Model} \label{sec:model}

%In a nutshell, we consider the standard CONGEST model, except that instead of an undirected graph the underlying graph is modeled by a {\em bidirected} graph, i.e. a directed graph in which the reverse of every edge is also an edge. This is because we have to deal with asymmetric edge weight (even when the initial network has symmetric weights).  Additionally, for simplicity we assume that nodes IDs are in the range of $\{0, 1, \ldots, n-1\}$. (This assumption can be achieved in $O(n)$ time.)

The communication network is modeled by an undirected unweighted  $n$-node $m$-edge graph $G$, where nodes model the processors and  edges model the {\em bounded-bandwidth} links between the processors. Let $V(G)$ and $E(G)$ denote the set of nodes and (directed) edges of $G$, respectively. 
The processors  (henceforth, nodes) are assumed to have unique IDs in the range of $\{0, 1, \ldots, n-1\}$ and infinite computational power. Typically nodes' IDs are assumed to be in the range of $\{1, \ldots, \poly(n)\}$. But as observed in \cite{HuangNS17}, in $O(n)$ time the range can be reduced to $\{0, 1, \ldots, n-1\}$. 
Each node has limited topological knowledge; in particular, it only knows the IDs of its neighbors and knows {\em no} other topological information (e.g., whether its neighbors are linked by an edge or not). 

Nodes may also accept some additional inputs as specified by the problem at hand. For the case of graph problems, the additional input is typically {\em edge weights}. Let $w:E(G)\rightarrow \{1, 2, \ldots, \poly(n)\}$ be the edge weight assignment.\footnote{Note that it might be natural to include $\infty$ as a possible edge weight. But this is not necessary since it can be replaced by a large weight of value $\poly(n)$.} We refer to network $G$ with weight assignment $w$ as the {\em weighted network}, denoted by $G(w)$. The weight $w(u,v)$ of each edge $(u,v)$ is known only to $u$ and $v$. 
%
%As commonly done in the literature, 
%%(e.g., \cite{KhanP08,LotkerPR09,KuttenP98,GarayKP98,GhaffariK13}), 
%we will assume that the maximum weight is $\poly(n)$; so, each edge weight can be sent through an edge (link) in one round.
%\footnote{We note that, besides needing this assumption to ensure that weights can be encoded by $O(\log n)$ bits, we also need it in the analysis of the running time of our algorithms: most running times of our algorithms are logarithmic of the largest edge weight. This is in the same spirit as, e.g., \cite{LotkerPR09,GhaffariK13,KhanP08}.}  
%
%We refer to the weight function as {\em symmetric}, or sometimes {\em undirected}, if for every (directed) edge $(u,v)$, $w(u,v)=w(v,u)$. Otherwise, it is called {\em asymmetric}, or sometimes {\em directed}. We note again that the symmetric case is the typical case considered in the literature, but we have to deal with the asymmetric case in our algorithm. 

We measure the performance of algorithms by its running time, defined as the worst-case number of {\em rounds} of distributed communication. 
At the beginning of each round, all nodes wake up simultaneously. Each node $u$ then sends an arbitrary message of $O(\log n)$ bits through each edge $(u,v)$, and the message will arrive at node $v$ at the end of the round. 
%ne being the {\em running time}, defined as the worst-case number of rounds of distributed communication. 
%
We assume that nodes always know the number of the current round for simplicity.
%
%To simplify notations, we will name nodes using their IDs, i.e. we let $V(G)\subseteq \{1, \ldots, \poly(n)\}$. Thus, we use $u\in V(G)$ to represent a node, as well as its ID. 
%
%The running time is analyzed in terms of number of nodes and edges ($n$ and $m$). 
%Often it is also analyzed in terms of  $\diam$, the diameter of the network $G$; this parameter will not appear in our running time. 
%
In this paper, the running time is analyzed in terms of the number of nodes ($n$). 
Since $n$ can be computed in $O(\diam)$ time, where $\diam$ is the diameter of $G$, we will assume that every node knows $n$.

\paragraph{Remark on edge weights and directions:} Note that our algorithm in fact works in the most restricted model studied in the literature, where edge weights are ``arbitrary'', edges are {\em directed}, and communications are {\em unidirectional}. 

It was commonly assumed in the literature (e.g., \cite{KhanP08,LotkerPR09,KuttenP98,GarayKP98,GhaffariK13,HuangNS17,GhaffariL18,ForsterN18}) that the maximum weight is $\poly(n)$; so, each edge weight can be sent through an edge (link) in one round. A more general ``arbitrary weight'' model has been considered in, e.g., \cite{Elkin-STOC17,AgarwalRKP18-PODC,AgarwalR18-Arxiv-3/2,AgarwalR18-Arxiv2-4/3}. In this model, edge weights can be arbitrary, and it is assumed that communication links have enough capacity to deliver a distance information in one round. Some algorithms do not work in this model, including the previously best $\tilde O(n^{5/4})$-time algorithm \cite{HuangNS17}. 

The case of directed graph has also been studied in the literature. One can consider further whether the communication is {\em bidirectional}, i.e. nodes can communicate on an edge regardless of its direction, or the more restricted {\em unidirectional} case, where the communication has to be done along edge directions. The previously best $\tilde O(n^{5/4})$-time algorithm \cite{HuangNS17} has to assume bidirectional communication.

% OLD VERSION BY CHIEN-CHUNG
%
%$G$ is the unweighted (possibly directed) $n$-node graph representing the network topology. 
%Each node has a unique ID. For simplicity, we assume the IDs are from 0 to $n-1$. 
%The diameter of $G$ is denoted by $D(G)$. Let $w :E(G)\rightarrow \mathbb{Z}_{\geq 0}$ denote the 
%non-negative integral weights of edges in $G$. We adopt the standard assumption that 
%the largest weight $w(e)$ is in the order of $\poly(n)$. For any weight assignment $w$, (i) let $G(w)$ be the weighted graph with topology $G$ and weight assignment $w$, (ii) for any nodes $s$ and $t$, $\dist_w(s, t)$ is the distance from $s$ to $t$ in $G(w)$. 

\subsection{Notation and Problem Definition}\label{sec:notations}

Let $G = (V,E)$ be a directed network with arbitrary non-negative weights: $V$ is the set of nodes, and $E$ the set of edges.
Let $n = |V|$ and $m = |E|$. Let $(u,v)$ denote the edge \emph{from} $u$ \emph{to} $v$,
and let $w(u,v)$ be the weight of this edge.
For every pair of nodes $s$ and $t$ in $G$, let $\dist(s, t)$ be the shortest distance from $s$ to $t$ in $G$. 
Note that since the underlying graph $G$ is directed, we might have $\dist(s,t) \neq \dist(t,s)$.
Let $\sp(s,t)$ refer to the shortest path from $s$ to $t$; if there are multiple such paths, choose one of the shortest paths with the minimal number of edges.
Let $\hop(s,t)$ be the number of edges on $\sp(s,t)$.

Throughout the algorithm, each vertex $v$ will maintain for every $u$ various distance estimates $d^v(u,v)$. 
When we refer to such estimates, the superscript $v$ will \emph{always} refer to the node that possesses this knowledge.

%\paragraph{Definition of the Problem:}

\begin{definition}[All-pairs shortest paths (APSP)]
An algorithm for distributed APSP must terminate with every vertex $v \in V$ knowing a value $d^v(u,v) = \dist(u,v)$, for every $u \in V$. 
\end{definition}

We now define a notion of {\em accuracy} for the local information at $v$.

\begin{definition}[$h$-hop-accurate]
	\label{dfn:hop-accurate}
For any positive integer $h$, We say that a distance estimate $d^v(u,v)$ is $h$-hop-accurate if the following holds: 
{\bf 1)} $d^v(u,v) \geq \dist(u,v)$ and {\bf 2)} if $\hop(x,y) \leq h$ then $d^v(u,v) = \dist(u,v)$. {\bf Note:} if $h \geq n$, then $h$-hop accuracy guarantees $d^v(u,v) = \dist(u,v)$.

%\begin{enumerate}
%\item $d^v(u,v) \geq \dist(u,v)$
%\item if $\hop(x,y) \leq h$ then $d^v(u,v) = \dist(u,v)$.
%Note that for $h \geq n$, $h$-hop accuracy guarantees that $d^v(u,v) = \dist(u,v)$.
%\end{enumerate}
\end{definition}

We say that an event holds {\em with high probability} (w.h.p.) if it holds with probability at least $1-1/n^c$, where $c$ is an arbitrarily large constant.

%\subsection{Initial Transformations}

%%%% NO LONGER NEED UNIQUE SHORTEST PATHS
%\paragraph{Unique shortest paths:} We will assume that that for every pair of nodes $s$ and $t$,
%there is a \emph{unique} shortest path $P^*(s, t)$ from $s$ to $t$. It is well known that one can enforce uniqueness by adding small random perturbations to the edge weights: for any constant $c$, we can ensure that shortest paths are unique with probability at least $1 -1/n^c$ while only increasing the number of bits needed to represent each edge weight by an additive $O(log(n^c)) = O(log(n))$. See e.g. Section 6 of \cite{CabelloCE13} for more details.

\subsection{Distributed Algorithmic Primitives}
\label{sec:basic-algorithms}

\paragraph{The Bellman-Ford Algorithm.} This well-known algorithm computes SSSP from a source $s$ on a network $G$. The algorithm runs for $h$ rounds, where $h$ is an input given by the user. The algorithm offers the following guarantee: upon termination, $d^t(s, t)$ is $h$-hop-accurate for every node $t$ in $V$. See \Cref{sec:appendix-bellman} for a brief description of the algorithm.

\paragraph{Scheduling of Distributed Algorithms.} Consider $k$ distributed algorithms $A_1, A_2  \dots , A_k$. Let \dilation be such that each algorithm $A_i$ finishes in \dilation rounds if it runs individually. Let \congestion be such that there are at most \congestion messages, each of size $O(\log n)$, sent through each edge (counted over all rounds), when we run all algorithms together. We note the following result of Ghaffari \cite{Ghaffari15-scheduling}:

\begin{theorem}[\cite{Ghaffari15-scheduling}]\label{thm:Ghaffari}
There is a distributed algorithm that can execute $A_1, A_2  \dots , A_k$ altogether in $O(\dilation+\congestion\cdot \log n)$ time.
\end{theorem}

\paragraph{Negative edge weights}: If the original graph has negative weights (and no negative-weight cycles), then we can use the idea of reduced weights from Johnson's algorithm \cite{Johnson77} to transform the graph into a new graph with non-negative edge weights that has the same shortest paths as the original graph. The transformation requires $O(n)$ rounds. We can thus assume for the rest of the paper that weights are non-negative. See \Cref{sec:negative-weights} for more details.

%%%%% CENTER SAMPLING

%\paragraph{Center Sampling}

%Here we use a well known lemma of Ullman and Yannakakis~\cite{UllmanY91},
%~\cite[Lemma~2.2]{UllmanY91}

%\begin{lemma}[\cite{UllmanY91}]
%Given any positive parameters $z$ and $p$, if each node in $V$ is sampled independently with probability 
%$pn/z$, then given some (acyclic) path $P$ of length at least $z$, the probability that \emph{none} of the nodes on $P$ is sampled is at most $n^{-\alpha q}$, for some positive $\alpha$.
%\end{lemma}

%Taking a union bound for all $n^2$ shortest paths in the graph, this lemma implies the following:

%\begin{lemma} 
%\label{lem:center-sampling}
%Given any parameter $z$, there exists a sufficiently large constant $q$ such that if each node in $V$ is sampled with
%probability $qn\log(n)/z$, then with high probability: there exists a sampled node on $P^*(x,y)$ for all pairs of nodes
%$x$ and $y$ for which $h(x,y) \geq z$.
%\end{lemma}

%%%%%%%%%%%% STUFF FROM OLD FILE

%%%%%%%%%%%%%%%%%%%%% CENTER SAMPLING

	\section{The Algorithm}

Define $S_0 = V$. Let $k = \log(n)$, and for each $i=0, \ldots, k$, select each node to $S_i$ with probability $(1/2)^i$ (every node knows whether it is in $S_i$ or not). Let $S_{k+1}=\emptyset$. (Note that we do not require $S_{i+1} \subseteq S_i$.) The following facts follow from standard techniques.

\begin{lemma}\label{thm:standard-bounds} W.h.p., the following holds for every $i$. 
	\begin{itemize}[noitemsep]
		\item  $|S_i|=O(n \log n / 2^i)$, and 
		\item for a large enough constant $c$ and for every pairs of nodes $u$ and $v$ such that $\hop(u, v)\geq c2^i\log n$, the shortest path $\sp(u,v)$ contains a node in $S_i$.
	\end{itemize}
\end{lemma}

\begin{algorithm}
	\caption{Phase $i$ of the Main APSP Algorithm}\label{alg:main-algo2}
	
	\KwIn{Every node $v$ knows the phase number $i$. For every pair of nodes $s\in \cup_{j\geq i+1} S_j$ and $v\in V$, $v$ knows $d^v_{i+1}(s,v)=\dist(s,v)$.}
	\KwOut{For every pair of nodes $s \in \cup_{j\geq i} S_i$ and $v \in V$, $v$ knows $d^v_{i}(s,v)=\dist(s,v)$.}

	%Starting from $i=\log n$ down to $0$, we will do the following. 
	
	Run Bellman-Ford from every node $s\in S_i$ up to depth $c2^{i+1}\log n$, for a large enough constant $c$.  Let $\distest^v_i(s, v)$ be the resulting distance estimate each node $v$ learns about $s\in S_i$ from this step.  Note that $\distest^v_i(s, v)$ is $(c2^{i+1}\log n)$-hop accurate. 
	\tcc{By \Cref{thm:standard-bounds}, this takes $O(|S_i|2^{i+1}\log n)= O(n\log^2 n)$ rounds w.h.p. } \label{line:main-bellman}
	
	%For every nodes $u\in S_i$ and $v\in V$ such that $\hop(u,v)\leq 2^i\log n$, make $v$ knows $\dist(u, v)$. This can be done by running Bellman-Ford from every node $u\in S_i$ for $c2^i\log n$ rounds, for a large enough constant $c$. By \Cref{thm:standard-bounds}, this takes $O(|S_i|2^i\log n)= O(n\log^2 n)$ rounds in total w.h.p.
	%\tcc{By \Cref{thm:standard-bounds}, this takes $O(|S_i|2^i\log n)= O(n\log^2 n)$ rounds in total w.h.p.}
	
	Call \filteredbroadcast($s$, $S_{i+1}$, $\{\distest^b_i(s, b)\}_{b\in S_{i+1}}$, $\{d^v_{i+1}(b, v)\}_{b\in S_{i+1}, v\in V}$) 
	%\filteredbroadcast($s$, $i$) 
	(cf. \Cref{alg:filtering2}) in parallel for every node $s\in S_i$, using Theorem \ref{thm:Ghaffari} to schedule the parallel agorithms. We show in \Cref{sec:filtering-correctness} that this makes every node $v\in V$ know  
	$$\bar{d}^v_i(s, v):=\min_{b\in S_{i+1}} \left(\distest^b_i(s, b)+d^v_{i+1}(b, v)\right) = \min_{b\in S_{i+1}} \left(\distest^b_i(s, b)+\dist(b, v)\right) \,.$$
	% and w.h.p. takes $\tilde O(n)$  XXXX 
	\label{line:main-filtering}
	
	Every node $v$ internally computes $d^v_i(s, v):=\min(d_{i+1}^v(s, v), \distest^v_i(s, v), \bar{d}^v_i(s, v))$ for every node $s\in \cup_{j\geq i} S_i$. \label{line:main-internal}

\end{algorithm}

%\begin{proof}
%	TO DO: First claim. 
%	
%	The probability that the second claim does not happen is $(1-p^{i+1})^{c2^i\log n}=O(e^{-\log n}) = O(1/n^{-c})$. 
%\end{proof}

%Make every node $v$ knows $\dist(u, v)$ for every $u\in S_k$; by running Bellman-Ford for $n$ rounds from every $u\in S_k$, this takes $\tilde O(n|S_k|)=\tilde O(n)$ time w.h.p. 

Our algorithm runs in {\em phases}, starting from $i=k$ down to $0$. In each phase, we execute \Cref{alg:main-algo2}. At the end of phase $i$, every $v \in V$ knows distance $d^v_{i}(s,v)=\dist(s,v)$ for every $s \in S_i$. Since $S_0 = V$, the algorithm terminates with knowledge of APSP.

In \Cref{alg:filtering2}, we describe the \filteredbroadcast algorithm. Note that in our main algorithm, the set of between-nodes $B$ is always equal to some $S_i$; but the \filteredbroadcast subroutine in fact works for an arbitrary set $B$, so we describe in its full generality.

\subsection{Correctness of the Main Algorithm (\Cref{alg:main-algo2})}

In this subsection we show that the output of phase $i$: for every pair of nodes $s \in \cup_{j\geq i} S_i$ and $v \in V$, $v$ knows $d^v_{i}(s,v)=\dist(s,v)$.
Since $d_{i+1}^v(s, v)=\dist(s, v)$ for every $s\in \cup_{j\geq i+1} S_j$, it is enough to show that $d_{i}^v(s, v)=\dist(s, v)$ for every $s\in S_i$. It is clear that $d_{i}^v(s, v) \geq \dist(s, v)$ because every distance returned by our algorithm corresponds to some path in the graph; we now complete the proof by showing that 
$d_{i}^v(s, v) \leq \dist(s, v)$.   

Consider any fixed pair $s\in S_i$ and $v\in V$, and let $c$ be the constant in Lemma \ref{thm:standard-bounds}. 

\paragraph{Case 1: $\hop(s, v)\leq c2^{i+1}\log n$.} Then in Step \ref{line:main-bellman},  $d_i^v(s,v) \leq \distest_i^v(s, v)=\dist(s, v)$ because by the properties of Bellman-Ford, $\distest_i^v(s, v)$ is $c2^{i+1}\log n$-hop-accurate (see \cref{dfn:hop-accurate}).

\paragraph{Case 2: $\hop(s, v)> c2^{i+1}\log n$.} By \Cref{thm:standard-bounds}, there exists a node $b\in S_{i+1}$ that is contained in the shortest path $\sp(s,v)$ and $\hop(s, b)\leq c2^{i+1}\log n$.  Bellman-Ford in Step \ref{line:main-bellman} ensures $\distest_i^b(s, b)=\dist(s, b)$, and since $b \in S_{i+1}$, we have  $d_{i+1}^v(b, v)=\dist(b, v)$ (by the input assumption for phase $i$). Thus: $
d_i^v(s,v) \leq \bar{d}_i^v(s, v)\leq \distest_i^b(s, b)+d_{i+1}^v(b, v)=\dist(s, v)$.

\subsection{Correctness of \filteredbroadcast (\Cref{alg:filtering2})}\label{sec:filtering-correctness}

\begin{lemma}\label{thm:filtering-correctness}
	For any $j\leq \log(n)+1$ and every node $v$, when Iteration $j$ terminates $v$ knows $output^v = \min_{b\in B_j} \distest^b(s, b)+d^v(b, v)$; here we define $\min_{b\in \emptyset}(\cdot)=\infty$.
\end{lemma}

\begin{algorithm}
	\caption{\filteredbroadcast($s$, $B$, $\{\distest^b(s, b)\}_{b\in B}$, $\{d^v(b, v)\}_{b\in B, v\in V}$)}\label{alg:filtering2} 	
	
	\KwIn{Source node $s$ ($s$ knows that it is the source and will initiate the excution of this algorithm). 
		Set $B$ of {\em between-nodes} (every node knows whether it is in $B$ or not). 
		Every node $b\in B$ knows a distance estimate $\distest^b(s, b)$. %where  $\dist(u,v)\leq \distest^u(s, u)$. 
		Every node $v\in V$ knows a distance $d^v(b, v)=\dist(b, v)$ for every $b\in B$. 
	}
	\KwOut{Every node $v$ has value $output^v=\min_{b\in B} \distest^b(s, b)+d^v(b, v)$.
		%		such that
		%		\begin{align}\label{eq:filtering-output}
		%		output^v \begin{cases}
		%		= \min_{b\in B} d_{i-1}^b(s, b)+d_{i}^v(b, v) & \mbox{if $\hop(s, v)\leq 2^i$, and}\\
		%		\geq \min_{b\in B} d_{i-1}^b(s, b)+d_{i}^v(b, v) & \mbox{otherwise.}
		%		\end{cases}
		%		\end{align}
		%%		$output^v=\min_{b\in B} d_i^b(s, b)+d_{i+1}^v(b, v)$ if $\hop(s, v)\leq 2^i$, and
		%%		$output^v\geq \min_{b\in B} d_i^b(s, b)+d_{i+1}^v(b, v)$ otherwise.
	}
	
	\tcc{Note that the algorithm requires $d^v(b, v)=\dist(b, v)$, but that the input estimates $\distest^b(s, b)$ can be arbitrary numbers.}
	
	\medskip
	
	Let $B_0=B$. For any $j=1, \ldots, \log n$, define $B_j$ to be a set where we select each node in $B_{j-1}$ to $B_j$ with probability $1/2$. Let $B_{\log n + 1}=\emptyset$. Note that every node $b\in B$ can decide (randomly) whether or not it is in each $B_j$ without any communication. 
	
	%Let $p^v$ be the maximum $j$ such that $v\in B_j$. We call $p^v$ the {\em priority} of $v$. Node that $v$ can compute $p^v$ without any communication. 
	
	Every node $v$ creates a variable $output^v=\infty$.  
	
	%	\label{line:filter-iteration-logn} Broadcast  $\distest^b(s, b)$ for every $b\in B_{\log n}$ to the whole network. 
	%	Every node $v$ creates a variable $output^v=\min_{b\in B_{\log n}} (\distest^b(s, b)+d^v(b, v))$. 
	
	\label{line:filter-iteration} Execute the following in {\em iterations}. Starting from Iteration $j=\log(n)$ down to Iteration $j=0$. Each iteration lasts $n$ rounds. In Iteration $j$, do the following.
	
	{
		\begin{enumerate}[noitemsep,label=(\roman*)] 
			\item \label{line:filtering-init} Round $1$: In parallel, every node $b\in B_j$ sets $output^b=\distest^b(s, b)$ and sends message $M(s, b)=(s, b, \distest^b(s, b))$ to all neighbors. 
			
			\tcc{The M(s,b) of round 1 constitute the entire message set of iteraion $j$; future rounds then determine how these messages are passed on.}
			
			\item At Rounds $2$ to $n$, every node $v$ does the following. \label{line:filtering-each-round}
			
			\begin{enumerate}[noitemsep,nolistsep]
				\item Let $B^v$ be the set of nodes $b$ such that $v$ has received the message $M(s, b)$.   
				%
				%Node $v$ may receive some messages of the form $(s, b_1, d_i^b(s, b_1)), (s, b_1, d_i^b(s, b_1)), \ldots$, for some $b_1, b_2, \ldots \in B.$ 
				%
				Let $b^*$ be the node in $B^v$ that minimizes $\distest^b(s, b)+d^v(b, v)$; ties can be broken arbitrarily. 
				%if there is a tie, select $b\in B$ with smallest ID as $b^*$.  
				
				\item \label{line:filter-condition} If $\distest^{b^*}(s, b^*)+d^v(b^*, v)<output^v$, then $v$ sends message  $M(s, b^*)$ to all neighbors and sets $output^v=\distest^{b^*}(s, b^*)+d^v(b^*, v)$. 
			\end{enumerate} 
			
		\end{enumerate}
	}	
\end{algorithm}

\begin{proof}
	It is clear that $output^v \geq \min_{b\in B_j} \distest^b(s, b)+d^v(b, v)$, because $v$ only considers values of the form $\distest^b(s, b)+d^v(b, v)$. The harder direction is to show that $output^v \leq \min_{b\in B_j} \distest^b(s, b)+d^v(b, v)$.

	We prove this by induction on $j$. The base case for $j = \log(n)  + 1$ is trivial, since we define $\min_{b\in \emptyset}(\cdot)=\infty$.
	For the induction step, we assume that the Lemma holds for  $j+1 \leq \log n + 1$, and will show that it holds for iteration $j$ as well. 
 
	Let us fix some particular node $v$. We now consider two cases; the first is much simpler.
	
	\paragraph{Case 1: $\arg\min_{b\in B_j} \distest^b(s, b)+d^v(b, v) \cap B_{j+1} \neq \emptyset$}. Intuitively, this is the case that the ``best" between-node for $v$ in $B_j$ is no better than the best node from $B_{j+1}$; since we know that Lemma \ref{thm:filtering-correctness} holds for iteration $j+1$ (inductive hypothesis), the assumption of Case 1 directly ensures that it also holds for iteration $j$ .
	
	\paragraph{Case 2: $\arg\min_{b\in B_j} \distest^b(s, b)+d^v(b, v) \cap B_{j+1} = \emptyset$}. The rest of the proof is concerned with this case. Let $u_0$ be any vertex in $\arg\min_{b\in B_j} \distest^b(s, b)+d^v(b, v)$, and say that $\sp(u_0,v) = (u_0, u_1, \ldots u_\ell=v)$, for some length $\ell$. 
	
	Because of the case assumption, we know that for any $b \in B_{j+1}$ we have $\distest^{u_0}(s, u_0)+d^v(u_0, v) < \distest^b(s, b)+d^v(b, v)$. Moreover, it is not hard to see that because $\sp(u_0,v)$ is a shortest path, we must also have 
	\begin{equation}
	\label{eq:filtering-correcntess}
	\distest^{u_0}(s, u_0)+d^{u_i}(u_0, u_i) < \distest^b(s, b)+d^{u_i}(b, u_i) \qquad \mbox{for every $b \in B_{j+1}$ and $u_i \in \pi(s,v)$.}
	\end{equation}
	
	Now, the intuition behind the proof of Lemma \ref{thm:filtering-correctness} is that $v$ should receive the message $M(s,u_0)$, which by choice of $u_0$ will ensure that $v$ sets $output^v \leq \min_{b\in B_j} \distest^b(s, b)+d^v(b, v)$.
	The reason we expect this message to travel all the way to $v$ is
	because by Equation \ref{eq:filtering-correcntess}, for every $u_i \in \sp(u_0,v)$, $u_0 \in B_j$ is a better between node for $u_i$ than all $b \in B_{j+1}$, so message $M(s,u_0)$ will pass the filter in Step \ref{line:filter-condition} of \cref{alg:filtering2}. The one issue with this proof is that some $u_i$ may fail to pass along $M(s,u_0)$ if it passed along an equally good message $M(s,b)$ in an earlier round of iteration $j$; but this is still fine, as we will show that because $\sp(u_0,v)$ is a shortest path, this message $M(s,b)$ is also good for $v$.

	\begin{claim}\label{thm:filter-good-message-received}
	For every $r$, by the end of Round $r$ of Iteration $j$, $u_r$ has received a message $M(s, b)$ for some $b$ such that 
	\begin{align}
	\distest^{b}(s, b)+d^{u_r}(b, u_r)\leq \distest^{u_0}(s, u_0)+d^{u_r}(u_0, u_r).\label{eq:filter-claim-message-condition}
	\end{align}
	%$b\in B^*_j(v)$. 	
	\end{claim} 
	\begin{proof}
%		First, observe that	$B_{u_0}\subseteq  B_{u_1} \subseteq \ldots \subseteq B_{u_\ell}$: 
%		
		The proof is by induction on the number of $r$. Node $u_0 \in B_j$ sends $M(s, u_0)$ in Round $1$ of Iteration $j$, so the claim is obviously true for $u_1$. We now assume assume that the claim is true for some $u_r$, with $r \geq 1$, and show that the claim must hold for $u_{r+1}$ as well.
		Let us consider the \emph{first} message $M(s,b)$ received by $u_r$ for which \Cref{eq:filter-claim-message-condition} is satisfied;
		by the induction hypothesis, $u_r$ receives this message at some time $t \leq r$, and moreover, by Equation \ref{eq:filtering-correcntess}, this event first occurs in iteration $j$; it could not have occurred in an earlier iteration $j' > j$. Thus, by Step~\ref{line:filter-condition} of \Cref{alg:filtering2}), we know that at time $t$, $u_r$ set $output^{u_r}$ to $\distest^{b}(s, b)+d^{u_r}(b, u_r)$, and sent message $M(s,b)$ to all its neighbors, including $u_{r+1}$

	Thus, $u_{r+1}$ receives message $M(s,b)$ at time $t+1 \leq r+1$. Observe that:
		\begin{align*}
		&\distest^{b}(s, b)+d^{u_{r+1}}(b, u_{r+1})\\
		&=	\distest^{b}(s, b)+\dist(b, u_{r+1}) &\mbox{(by input condition of \Cref{alg:filtering2})}\\
		&\leq \distest^{b}(s, b)+\dist(b, u_{r})+ \dist(u_r, u_{r+1}) &\mbox{(by triangle inequality)}\\
		&= \distest^{b}(s, b)+d^{u_r}(b, u_{r})+ \dist(u_r, u_{r+1}) &\mbox{(by input condition of \Cref{alg:filtering2})}\\
		&\leq \distest^{u_0}(s, u_0)+d^{u_r}(u_0, u_r)+ \dist(u_r, u_{r+1})  &\mbox{(by \Cref{eq:filter-claim-message-condition})}\\
		&= \distest^{u_0}(s, u_0)+\dist(u_0, u_r)+ \dist(u_r, u_{r+1})  &\mbox{(by input condition of \Cref{alg:filtering2})}\\
		&= \distest^{u_0}(s, u_0)+\dist(u_0, u_{r+1}) &\mbox{(since $u_r$ is on the shortest $(u_0u_{r+1})$-path)}\\
		&\leq \distest^{u_0}(s, u_0)+d^{u_{r+1}}(u_0, u_{r+1}) &\mbox{(by input condition of \Cref{alg:filtering2})}.
		\end{align*}
		Message $M(s, b)$ thus satisfies \Cref{eq:filter-claim-message-condition} for node $u_{r+1}$, which completes the induction proof of \Cref{thm:filter-good-message-received} 
	\end{proof}
Since each iteration has $n$ rounds, \Cref{thm:filter-good-message-received} implies that node $v$ will receive $M(s, b)$ satisfying \Cref{eq:filter-claim-message-condition} by the end of Iteration $j$; thus by the choice of $u_0$, Step~\ref{line:filter-condition} of \Cref{alg:filtering2}) sets $output^v$ to be at most 
$$\distest^{b}(s, b)+d^{u_r}(b, u_r)\leq \distest^{u_0}(s, u_0)+d^{u_r}(u_0, u_r)=\min_{b\in B_j} \distest^b(s, b)+d^v(b, v,)$$ 
as desired. This concludes the proof of \Cref{thm:filtering-correctness}. 
\end{proof}

\begin{observation}
	\label{obs:h}
	 We do not need this for our main result, but we note that if the input had the additional guarantee that all hop-distances between $B$ and $V$ were at most $h$, then we would only need to run \filteredbroadcast for $h$ rounds instead of $n$; this is because in case 2 of the proof, $h$ rounds would suffice for the message to propagate from $u_0$ to $u_\ell = v$.
\end{observation}

\subsection{Complexity of \filteredbroadcast (\Cref{alg:filtering2})}

The time complexity of \Cref{alg:filtering2} is clearly $O(n\log(n))$, since there are $O(\log(n))$ iterations, and each is specified to run for $O(n)$ rounds. Now we show that the algorithm creates low congestion on every edge, and thus can be easily parallelized.  

\begin{lemma}\label{thm:filter-complexity}
W.h.p., every node $v$ sends to its neighbors $O(\log n)$ messages of the form $M(s, b)$ in each iteration $j$ of \Cref{alg:filtering2}. 
\end{lemma}
\begin{proof}
	The claim is true for $j=\log n$ because by \cref{thm:standard-bounds}, $|B_{\log n}|=O(\log n)$ w.h.p.
	Now, Consider any Iteration $j<\log n$ of \Cref{alg:filtering2} and any node $v$. 
	Define for any node $b\in B$  
	$$B'_j(b, v)=\{b'\in B_j|  \distest^{b'}(s, b')+d^v(b', v)\leq  \distest^b(s, b)+d^v(b, v)\}.$$
	Observe that if $B'_j(b, v) \cap B_{j+1}\neq \emptyset$, 
	then $v$ will not send $M(s, b)$ to neighbors in Step \ref{line:filter-condition} of Iteration $j$,
	because by \cref{thm:filtering-correctness}, 	at the end of the previous iteration ($j+1$) we will already have 
	 $output^v = \min_{b' \in B_{j+1}} (\distest^b(s, b')+d^v(b', v))\leq  \distest^b(s, b)+d^v(b, v)$. 
	Observe further that the definition of $B'_j(b, v)$ does not depend on the randomness used to sample $B_{j+1}$; thus, since each $b \in B_j$ is sampled into $B_{j+1}$ with probability $1/2$, we have:
	$$\forall \gamma>0, \forall |B'_j(b, v)|\geq \gamma\log n,  Pr[ B'_j(b, v) \cap B_{j+1}=\emptyset]=1/2^{|B'_j(b, v)|}=O(n^{-\gamma}),$$
	
	Applying a union bound over the $\leq n$ possible values of $b$, we get: w.h.p, for all messages $M(s,b)$ sent by $v$ with $b \in B_j$, we have that $|B'_j(b, v)|=O(\log n)$.   \Cref{thm:filter-complexity} follows from the fact that there are $O(\log n)$ nodes $b \in B_j$ with $|B'_j(b, v)|=O(\log n)$. To see this, order nodes $b$ in $B_j$ by increasing values of $\distest^b(s, b)+d^v(b, v)$ (break ties arbitrarily). Observe that for the $i^{th}$ node $b$ in this order, $B'_j(b, v)$ contains all nodes that appear before $b$ in the order. Thus, only the first $O(\log(n))$ in this order have the property that $|B'_j(b, v)|=O(\log n)$. 
\end{proof}

\begin{corollary}
	\label{cor:filter-complexity}
	Over all $\log(n)$ iterations, 
	\cref{alg:filtering2} terminates in $O(n\log(n))$ rounds, and incurs a congestion of $O(\log^2(n))$ on each edge.
\end{corollary}

%Finally, athough we do not need this for our main result, we note that the bound of $O(\log^2(n))$ congestion per edge, combined with Observation \ref{obs:h}, implies that if all hop-distances from $B$ to $V$ were guaranteed to be at most $h$, then running  \filteredbroacast from $\beta$ different sources $s$ (using Theorem \ref{thm:Ghaffari}) would a total complexity of $\tilde{O}(\beta + h)$.

\subsection{Complexity of the Main Algorithm (\Cref{alg:main-algo2})}

Recall that \Cref{alg:main-algo2} runs in $k=\log(n)$ phases. We now analyze the complexity of an individual phase: summing over all the phases completes the proof of our main result.

\begin{lemma}
	\label{lem:main-complexity}
W.h.p phase $i$ of \Cref{alg:main-algo2} terminates in 
$O(n\log^2(n) + |S_i|\log^3(n)) = O(n\log^2(n) + n\log^4(n)/2^i)$ rounds
\end{lemma}

\begin{proof}
By \Cref{thm:standard-bounds} and the properties of Bellman-Ford, Step \ref{line:main-bellman} of \Cref{alg:main-algo2} requires a total of $O(|S_i|2^{i+1}\log n)= O(n\log^2 n)$ rounds w.h.p.

Step \ref{line:main-internal} of \Cref{alg:main-algo2} does not require any communication, so all that remains is to analyze the number of rounds required for all the calls to \filteredbroadcast\ in Step \ref{line:main-filtering}. The algorithm runs $|S_i|$ instances of \filteredbroadcast\ in parallel. By corollary \ref{cor:filter-complexity} each runs in $O(n\log(n))$ rounds and incurs $O(\log^2(n))$ congestion per edge. Thus the total congestion is $O(|S_i|\log^2(n))$, so using the parallel scheduler in \cref{thm:Ghaffari} yields a total round complexity of $O(|S_i|\log^3(n))$, as desired.
\end{proof}

%%%%%%%%%%%%%%%%
%%%%%%%%%%%%%%%%
%%%%%%%%%%%%%%%%

	%\input{filtering}
%	\input{filtering-OLD}
	
	\section{Open Problems}\label{sec:open}

As mentioned earlier, deterministic $\tilde O(n)$-time algorithms for APSP remains a key open problem. Additionally, while APSP admits an $\tilde \Omega(n)$ lower bound, it is a curious question whether this bound also holds for the following {\em strongly connected component} problem: We want every node to output a ``label'' such that two nodes are in the same strongly-connected component  if and only if their labels are the same (or even simpler, just counting the number of connected components). It should also be interesting to see how our algorithm performs in real systems (such as D-Galois \cite{HoangPDGYPR19}), and to see if our ideas are useful in computing various centrality measures (e.g.  \cite{HoangPDGYPR19}).

A few problems remain open for SSSP. An obvious one is closing the gap between  lower and upper bounds for SSSP \cite{ForsterN18,DasSarmaHKKNPPW12} and the single-source reachability problem \cite{GhaffariU15}.  Another question is whether the best upper bound for SSSP can be obtained without the scaling technique, so that we can avoid the dependency on the ratio between the highest and lowest edge weights. Recall that the previous state-of-the-art algorithms for both APSP and SSSP \cite{HuangNS17,GhaffariL18,ForsterN18} require this technique, but our algorithm does not.

This paper is part of an effort to understand {\em exact} distributed graph algorithms, and more generally to classify complexities of global problems in the CONGEST model. Many problems are yet to be settled, including minimum cut \cite{DagaHNS19}, maximum weight/cardinality matching \cite{AhmadiKO18-matching}, st-cut/flow \cite{GhaffariKKLP15}, vertex connectivity \cite{Censor-HillelGK14-decomposition}, and densest subgraph \cite{DasSarmaLNT12}. As mentioned earlier, settling the exact cases for other problems remains a major open problem. As mentioned in  \cite{Censor-HillelKP17,DagaHNS19}, tight bounds witnessed so far are in the form of either $\tilde \Theta(D)$, $\tilde \Theta(\sqrt{n}+D)$, $\tilde \Theta(n)$, or $\tilde \Theta(n^2)$. Any tight bound in-between is of our interest.

Finally, we propose studying the relationship between the {\em node-partition} two-party communication complexity and distributed graph algorithms in the CONGEST model. The only known technique to prove a lower bound of $\tilde \Omega(t)$ for any $t\geq n$ (e.g. \cite{FrischknechtHW12,AbboudCK16,Censor-HillelKP17,Nanongkai-STOC14}) in the CONGEST model is to partition nodes into two sides and argue (via two-party communication complexity-theoretic arguments) that there must be $\tilde \Omega(t|C|)$ bits of information between the two sides, where $C$ is the set of edges between the two sides. (See, e.g., \cite{Censor-HillelKP17} for details.) Is this the only technique for proving superlinear lower bounds? 
%Given that this technique was shown incapable of proving a superlinear lower bound for APSP \cite{Censor-HillelKP17} before a near-linear upper bound is shown in this paper, it might 
%
In particular, experiences from APSP (where this technique was shown incapable of proving a superlinear lower bound \cite{Censor-HillelKP17} before we settle a near-linear upper bound here) make it tempting to conjecture that if there is a protocol $\cal A$ with $\tilde O(t|C|)$ total communication for solving any graph problem $P$ in the two-party model above, then there is an $\tilde O(t)$-time CONGEST algorithm $\cal B$ for $P$, for any $t\geq n$. This conjecture sounds too good to be true in general. It will be extremely exciting already if it holds for some natural class of graph problems, even just for some $t\in(\sqrt{n}, n^2)$.
%
%While this conjecture sounds too good to be true, there is no evidence against it at the moment, even for some $t>\sqrt{n}$.
A related, more plausible, conjecture is to consider when $\cal A$ takes only $\tilde O(t)$ rounds. 

%Partition nodes into two sets, denoted by $V_A$ and $V_B$.  There are two players, Alice and Bob, who know the information about all edges incident to $V_A$ and $V_B$, respectively. Can Alice and Bob compute the value of the minimum cut of $G$ by communicating $\tilde O(n^{1-\epsilon}|C|)$ bits? A negative answer to this question would imply a lower bound in the CONGEST model by a standard technique (e.g. \cite{FrischknechtHW12,AbboudCK16,Censor-HillelKP17,Nanongkai-STOC14}). A positive answer would rule out pretty much the only known technique to prove lower bounds and might lead to a fast algorithm in the CONGEST model, as happened for all-pairs shortest paths \cite{Censor-HillelKP17,BernsteinN19}.

	\section{Acknowledgement} 
	This project has received funding from the European Research Council (ERC) under the European Union's Horizon 2020 research and innovation programme under grant agreement No 715672. Nanongkai was also partially supported by the Swedish Research Council (Reg. No. 2015-04659.)

	% CITATIONS
	\ifdefined\AdvCite
	
	\printbibliography[heading=bibintoc] % Make bibliography show up in table of contents
	
	\else
	
	% MAKE SURE TO ADD BIB SOURCE IN TWO PLACES
	\bibliographystyle{plain}
	\bibliography{references}
	
	\fi
	
	\appendix
	
\section*{Appendix} \label{sec:appendix}

\section{Bellman-Ford}
\label{sec:appendix-bellman}

Since it figures prominently in our main algorithm, we now describe the  well-known Bellman-Ford algorithm for computing SSSP from a source $s$ on network $G$ \cite{Bellman58,Ford56}. We omit the analysis of the algorithm, since it can be found in the citations. The algorithm runs for $h$ rounds, where $h$ is an input given by the user.

For any node $t$, let $d^t(s,t)$ denote the knowledge of $t$ about 
$\dist(s, t)$. Initially, $d^t(s, t)=\infty$ for every node $t$, except that $d^s(s, s)=0$. The algorithm proceeds as follows. 
\begin{enumerate}[noitemsep,label=(\roman*)] 
	\item In round 0, every node $t$ sends $d^t(s, t)$ to all its neighbors.
	\item When a node $t$ receives the message about $d^x(s, x)$ from its neighbors $x$, it uses the new information to decrease the value of $d^t(s, t)$ if $d^x(s, x) + w(x,t) < d^t(s,t)$. 
	\item If $d^t(s, t)$ decreases, then node $t$ sends the new value of $d^t(s, t)$ to all its neighbors.
	\item Repeat (ii) and (iii) for $h$ rounds. 
\end{enumerate}

Clearly, the above algorithm takes $O(h)$ rounds. Moreover, it can be proved that when the algorithm terminates 
$d^t(s, t)$ is $h$-hop-accurate for every node $t$ in $V$.

\section{Non-negative weights}
\label{sec:negative-weights}

\newcommand{\computephi}{{\sc compute}-$\phi$}

In this section we show that if the original graph has negative weights but no non-negative cycles, we can in $O(n)$ rounds transform it to a graph that has exactly the same shortest path structure, but has non-negative weights. This justifies the assumption of non-negative weights in \cref{sec:prelim}. (If the graph has a negative cycle, then the algorithm will discover this cycle within $O(n)$ rounds.)

Our transformation directly follows the technique of reduced costs used in Johnson's APSP algorithm in the static setting \cite{Johnson77}. The algorithm will compute a node value $\phi(v)$ for every node $v$ such that the following property
is satisfied: for every edge $(x,y) \in E$, $\phi(x) + w(x,y) -\phi(y) \geq 0$. We show how to compute the values $\phi(v)$ later. Once these values are computed, the algorithm creates a new edge-weight function $w':E \rightarrow \mathbb{R}_{\geq 0}$, where $w'(x,y) = \phi(x) + w'(x,y) - \phi(y)$. 
Let $G' = (V,E')$ the graph with the weight function $w'$ instead of $w$, and let $\dist'(s,t)$ be the shortest $s-t$ distance in $G'$. It is to easy to see that $G'$ satisfies the following properties:
\begin{enumerate}
\item $w'(u,v) \geq 0$ for every edge $(u,v)$.
\item for every pair of nodes $s$ and $t$ we have $\dist(s,t) = \dist'(s,t) + \phi(t) - \phi(s)$.
\end{enumerate}

Thus overall algorithm proceeds as follows. First it executes process \computephi, described below: at the end of
this process, each vertex $v$ knows its own value $\phi(v)$. Then each vertex $v$ broadcasts $\phi(v)$ to the entire graph:
by \cref{lem:broadcast}, this takes a total of $O(n)$ rounds. 

%\paragraph{Broadcasting.} We need to follow fact following from basic upcasting and downcasting techniques \cite{Peleg00_book}. (The statement is from \cite{LenzenP_podc13}.)

\begin{lemma}[Broadcasting \cite{Peleg00_book}]\label{lem:broadcast}
Suppose each $v\in V$ holds $k_v \geq 0$ messages of $O(\log n)$ bits each, for a total of $K = \sum_{v\in V} k_v$ messages.
Then all nodes in the network can receive these $K$ messages within $O(K + \diam)$ rounds.
\end{lemma}

The algorithm then executes the main distributed APSP algorithm described in this paper on on $G'$ instead of $G$: by Property 1 of $G'$ it only encounters non-negative weights, as desired. When the APSP algorithm on $G'$ terminates, the guarantee is that for every pair of nodes $s$ and $t$, node $t$ knows $\dist'(s,t)$. By Property 2 of $G'$, $t$ can then figure out $\dist(s,t)$ using its knowledge of $\dist(s,t)$, $\phi(s)$ and $\phi(t)$.

All we have left to show is how to execute \computephi. Let the graph $G^*$ be the original graph $G$, but with an additional vertex $s^*$, and a directed edge of weight $0$ from $s^*$ to every node $v$. The algorithm then sets $\phi(v) = \dist(s^*, v)$ for every $v$. It is not hard to check that because of the triangle inequality for shortest distances, we have
$\phi(x) + w(x,y) -\phi(y) \geq 0$, as desired. We can compute $\dist(s^*, v)$ for every vertex $v$ by simply running Bellman-Ford for $n$ rounds: to deal with the fact that vertex $s^*$ does not actually exist, the algorithm executes Bellman Ford exactly as described in \cref{sec:basic-algorithms}, except that in step i) it initializes 
$\dist(s^*,v) = 0$ for every node $v$. Note that Bellman-Ford will also detect if there exists a negative weight cycle in the graph.

%	\appendix
%	\input{LowerBoundFail}

\end{document}